\documentclass[twocolumn]{autart}    % Enable this line and disable the 
\usepackage{amsmath,amsfonts,amssymb,mathrsfs,graphicx,cite,newproof}
\usepackage{multirow,url,subfigure}
\usepackage{color}

\newcommand{\tr}{\text{tr}}

\newtheorem{definition}{Definition}
\newtheorem{proposition}{Proposition}
\newtheorem{remark}{Remark}

\newtheorem{lemma}{Lemma}
\newtheorem{theorem}{Theorem}
\newtheorem{corollary}{Corollary}

\begin{document}

\begin{frontmatter}
%\runtitle{Insert a suggested running title}  % Running title for regular 
                                              % papers but only if the title  
                                              % is over 5 words. Running title 
                                              % is not shown in output.
{
\title{Performance and Robustness Analysis of Stochastic Jump Linear Systems using Wasserstein metric} % Title, preferably not more 
                                                % than 10 words.
}
\thanks[footnoteinfo]{Corresponding author K.Lee. This research was supported by National Science Foundation award \#1016299, with Dr. Helen Gill as the program manager.}
\author[rvt]{Kooktae Lee}\ead{animodor@tamu.edu},
\author[rvt]{Abhishek Halder}\ead{ahalder@tamu.edu},
\author[rvt]{Raktim Bhattacharya}\ead{raktim@tamu.edu}

\address[rvt]{Department of Aerospace Engineering, Texas A\&M University, College Station, TX 77843-3141
USA.}  % Please supply                                              

\begin{keyword}                           % Five to ten keywords,  
Performance and robustness analysis, stochastic jump linear systems, switched linear systems, Wasserstein distance.
\end{keyword}                             % keyword list or with the 
                                          % help of the Automatica 
                                          % keyword wizard

\begin{abstract}                          % Abstract of not more than 200 words.
{
%This paper focuses on the performance and the robustness analysis of stochastic jump linear systems. We present new results that enable uncertainty quantification for general stochastic jump linear systems, not necessarily for Markovian process. The performance and the robustness of the system is quantified via Wasserstein metric that assesses the distance between probability density functions. Both the transient and steady-state performance of the systems with given initial state uncertainties can be analysed in this framework. In addition, we prove that the convergence of this metric implies the mean square stability. Consequently, this paper provides a unifying framework for the performance and the robustness analysis of general stochastic jump linear systems. The practical usefulness and efficiency of the proposed method are verified through numerical examples.
This paper focuses on the performance and the robustness analysis of stochastic jump linear systems. The state trajectory under stochastic jump process becomes random variables, which brings forth the probability distributions in the system state. Therefore, we need to adopt a proper metric to measure the system performance with respect to stochastic switching. In this perspective, Wasserstein metric that assesses the distance between probability density functions is applied to provide the performance and the robustness analysis. Both the transient and steady-state performance of the systems with given initial state uncertainties can be measured in this framework. Also, we prove that the convergence of this metric implies the mean square stability. Overall, this study provides a unifying framework for the performance and the robustness analysis of general stochastic jump linear systems, but not necessarily Markovian jump process that is commonly used for stochastic switching. The practical usefulness and efficiency of the proposed method are verified through numerical examples.
}
\end{abstract}

\end{frontmatter}

\section{Introduction}
{
A jump linear system is defined as a dynamical system constructed with a family of linear subsystem dynamics and a switching logic that conduct a switching between linear subsystems. Over decades, jump linear systems have attracted a wide range of researches due to its practical implementations. For instance, jump linear systems are used for power systems, manufacturing systems, aerospace systems, networked control systems, etc.
In general, a jump linear system can be divided into two different categories depending on the switching logic. One branch is a deterministic switching where the jump process is deterministically given to the system. The utilization of such deterministic jump linear systems stems from plant stabilization\cite{minto1991new}, adaptive control\cite{narendra1994improving}, system performance\cite{lin2009stability}, and resource-constrained scheduling\cite{boctor1990some}. In most cases, the system stability has been one of the major issues to investigate since even stable subsystems make the system unstable by the switching. Hence, numerous results have been established for the stability analysis and the recent literature regarding the stability of deterministic jump linear systems can be found in \cite{lin2009stability}. In \cite{lin2009stability}, a sufficient condition for the stability of deterministic jump linear systems is guaranteed by solving certain linear matrix inequalities (LMIs). Also, the necessary and the sufficient conditions for the stability are shown via a finite tuple, satisfying a certain condition.

Unlike the deterministic jump linear system, a stochastic jump linear system (SJLS) that is another category of jump linear systems refers to systems with the stochastic switching process.
This type of jump linear systems is commonly used to represent the randomness in the switching such as communication delays or packet losses in the networked control systems\cite{hassibi1999control,xiao2000control}.
In \cite{hassibi1999control}, the networked control system with packet loss was modeled as an asynchronous dynamical system incorporating both discrete and continuous dynamics, and its stability was analysed through Lyapunov techniques. Since then, this problem has been formulated in a more general setting by representing the various aspects of communication uncertainties as Markov chains \cite{you2011minimum,coviello2011stabilization,xiong2006stabilization,xiong2007stabilization,liu2009stabilization}. Stability analysis in the presence of such uncertainty, has been performed in the Markov jump linear systems (MJLSs) framework \cite{xiao2000control, zhang2005new, wu2007design, zhang2009stability, karan2006transition,lee2006uniform}. Further, the stochastic stability for a class of nonlinear stochastic systems with semi-Markovian jump parameters is introduced in \cite{hou2006stochastic,li2013stochastic}. 
%Another example of stochastic jump systems is semi-Markov jump systems and the stochastic stability for these systems are introduced in \cite{hou2006stochastic,li2013stochastic}. 
Most previous literatures, however, have only dealt with steady-state analysis in terms of system stability.

Beyond the current literature, this paper has a key contribution for the analysis of a SJLS as follows.
Based on the theory of optimal transport \cite{villani2008optimal}, we propose new probabilistic tools for analysing the performance and the robustness of SJLSs. Compared to the current literatures that only guarantees asymptotic performance with a deterministic arbitrary initial state condition, our contribution is to develop a unifying framework enabling both transient and asymptotic performance analysis with uncertain initial state conditions. 
The main difficulty dealing with analysis of SJLSs is that the system trajectories differ from every run due to the random switching. Moreover, the system state becomes random variables with a probability density function (PDF) even with deterministic initial state conditions. Consequently, we need to adopt a proper metric to measure the performance and the robustness of SJLSs in the distributional sense. In this paper, the Wasserstein metric that enables quantification of the uncertainty is employed for the performance measure. We prove that the convergence of this metric implies the mean square stability.
To sum up, this paper provides the robustness analysis tools under the stochastic jumps with given initial state uncertainties without assuming any structure (e.g. Markov) on the underlying jump process.

The remainder of this paper is organized as follows. In Section II, we provide a brief review of the preliminaries. Section III deals with the performance and the robustness analysis of stochastic jump systems and develops computationally efficient tools for uncertainty quantification. Numerical examples are provided in Section IV, to illustrate the performance and the robustness analysis results developed in this work. Section V concludes the paper.

\textbf{Notation:} The set of real and natural numbers are denoted by $\mathbb{R}$ and $\mathbb{N}$, respectively. Further, $\mathbb{N}_{0} \triangleq \mathbb{N} \cup \{0\}$. The symbols $\tr\left(\cdot\right)$, $\otimes$, and $vec$ denote the trace of a square matrix, Kronecker product, and vectorization operators, respectively. The abbreviation m.s. stands for the convergence in mean-square sense. The notations $\mathbb{P}(\cdot)$ and $X \sim \rho\left(x\right)$ denote the probability and the random variable $X$ with PDF $\rho\left(x\right)$, respectively. The symbol $\mathcal{N}\left(\mu,\Sigma\right)$ is used to denote the PDF of a Gaussian random vector with mean $\mu$ and covariance $\Sigma$.}

\section{Preliminaries}
{
Consider a discrete-time jump linear system as follows.
\begin{align}
x(k+1) &= A_{\sigma_{k}}x(k), \quad k \in \mathbb{N}_{0} \label{dtMJLS}
\end{align}
\noindent where $x(k)$ is the state vector and $A_{\sigma_k}$ denotes the system matrices. $\sigma_k\in\mathcal{M}\triangleq\{1,2,\hdots,m\}$ stands for the stochastic jump process, governing the switching among $m$ different modes of (\ref{dtMJLS}). 
}

In this paper, we will consider general stochastic jump processes $\sigma_k$, and hence $\sigma_k$ can be any arbitrary random process. Then, the resulting dynamics becomes a SJLS as defined next.

%\subsection{Stochastic Jump Linear Systems (SJLS)}
\begin{definition} (\textbf{Stochastic jump linear system})
Tuples of the form $(\pi(k),A_{\sigma_{k}}(x(k)),\mathcal{M})$ is termed as a SJLS, provided the mode dynamics are given by (\ref{dtMJLS}); $\pi(k)$ denote the time-varying occupation probability vectors for prescribed stochastic processes $\sigma_{k}$.
\label{SJSdefn}
\end{definition}

\begin{remark}
A SJLS, as defined above, is a collection of modal vector fields and a sequence of mode-occupation probability vectors. If the jump processes $\sigma_{k}$ is deterministic, then at each time, $\pi(k)$ will have integral co-ordinates (single 1 and remaining $m-1$ zeroes), resulting in a \emph{deterministic switching sequence}. If, however, $\sigma_{k}$ is stochastic jump processes, then $\pi(k)$ will contain proper fractional co-ordinates, resulting in a \emph{randomized switching sequence} where at each time, exactly one out of $m$ modes will be chosen according to probability $\pi(k)$. Thus, starting from a deterministic initial condition, each execution of the SJLS may result in different switching sequences corresponding to random sample paths of $\sigma_{k}$ over $\mathcal{M}$. Every realization of these random switching sequences results in a trajectory realization on the state space, and hence repeated the SJLS executions, even with a fixed initial condition, yields a spatio-temporal evolution of joint state PDF: $\rho\left(x\left(k\right)\right)$. 
\label{GeneralitySJS}
\end{remark}
{
According to the structure that governs the temporal evolution of $\pi(k)$, some subsets of the stochastic jump processes can be listed as follows.

\begin{enumerate}
\item[1)] i.i.d. jump process:\\
A SJLS switching sequence is called stationary, if the occupation probability vector $\pi\left(k\right)$ remains stationary in time. In particular, a stationary deterministic switching sequence implies execution of a single mode (no switching). A stationary randomized switching sequence implies i.i.d. jump process.\\

\item[2)] Markov jump process:\\
Consider a discrete-time discrete state Markov chain with mode transition probabilities given by
\begin{align}
p_{ij} = \mathbb{P}\left(\sigma_{k+1}=j\mid\sigma_{k}=i\right)
\label{TransitionProb}
\end{align}
where $p_{ij} \geq 0$, $\forall i,j\in\mathcal{M}$. Hence, for $k\geq0$, the probability distribution $\pi\left(k\right) \in \mathbb{R}^{m}$ of the modes of (\ref{dtMJLS}), is governed by
\begin{eqnarray}
\pi(k+1) = \pi(k)P, \quad \pi(0) = [\pi_{1}(0)\ \cdots\ \pi_{m}(0)]
\label{piEqndtMJLS}
\end{eqnarray}
where the \emph{transition probability matrix} $P \in \mathbb{R}^{m\times m}$ is a right stochastic matrix with row sum $\sum_{j=1}^{m} p_{ij} = 1$, $\forall i\in\mathcal{M}$.\\

\item[3)] semi-Markov jump process:\\
For a homogeneous and discrete-time semi-Markov chain, semi-Markov kernel $q$ is defined by
\begin{align}
q_{ij}(k) = \mathbb{P}(\sigma_{n+1} = j, X_{n+1}=k | \sigma_n = i )
\end{align}
where $X_n$ denotes the sojourn time in state $\sigma_n=i$. Note that the transition probability $p_{ij}$ in Markov chain can be expressed in terms of the semi-Markov kernel by $p_{ij} = \sum_{k=0}^{\infty}q_{ij}(k)$.

%\item[4)] Poisson jump process:\\
%The jump linear system is referred to as a Poisson jump linear system if the mode transition probability satisfies $p_{ij}=p_{j}$, $\forall i,j\in\mathcal{M}$, where $p_{ij}$ is defined as \eqref{TransitionProb}. Hence, Bernoulli jump linear system is a particular case of MJLSs.

\end{enumerate}
A SJLS refers to the jump linear system for which jump process $\sigma_k$ is governed by any stochastic probability distribution $\pi(k)$. Consequently, a SJLS implies the jump linear system, where the jump probability distribution $\pi(k)$ forms proper fractional numbers with any arbitrary updating rules for $\pi(k)$.
}

%\begin{remark}(\textbf{Switching rule and switching sequence: nomenclature})
%In this paper, we use the phrases ``switching sequence" and ``switching rule" interchangeably, to mean \emph{time-dependent} switching protocol. Some papers (e.g. \cite{lee2011joint}) make a distinction between the two by reserving the nomenclature ``switching rule" for \emph{state-dependent} protocol while ``switching sequence" denotes \emph{time-dependent} protocol. This distinction stems from a discrete linear inclusion counterexample due to Stanford and Urbano \cite{stanford1994some}, where the system admits a stabilizing state-dependent switching protocol but no time-dependent stabilizing switching sequence. Since this paper exclusively deals with time-dependent protocols, we have no scope for such confusions.
%\label{SequenceVersusRule}
%\end{remark}

{\section{Performance and Robustness Analysis using Wasserstein metric}}
%\section{Robustness to Switching and Initial State Uncertainties}
\label{Perf}
Uncertainties in a SJLS appear at the execution level due to random switching sequence. Additional uncertainties may stem from imprecise setting of initial conditions and parameter values. These uncertainties manifest as the evolution of $\rho\left(x\left(k\right)\right)$. Thus, a natural way to quantify the uncertainty in the performance of a SJLS, is to compute the ``distance" of the instantaneous state PDF from a reference measure. In particular, if we fix the reference PDF as Dirac delta function at the origin, denoted as $\delta\left(x\right)$, then the time-history of this ``distance" would reveal the rate of convergence (divergence) for the stable (unstable) SJLS in the distributional sense.

For meaningful inference, the notion of ``distance" must define a metric, and should be computationally tractable. The choice of the metric is very important as it must be able to highlight  properties of density functions that are important from a dynamical system point of view. We propose that  the shape of  the density functions characterizes the dynamics of the system. Regions of high probability density correspond to high likelihood of finding the state there, which corresponds to higher concentration of trajectories. Higher concentration occurs in regions with low time scale dynamics or time invariance. For example, for a stable system, all trajectories accumulate at the origin and the corresponding PDF is the Dirac delta function at the origin. Similarly, low concentration areas indicate fast-scale dynamics or instability, and the corresponding steady-state density function is zero in the unstable manifold. Therefore, behavior of two dynamical systems are identical in the distribution sense if their state PDFs have identical shapes. {In order to properly capture the above aspects in dynamical systems, we adopt Wasserstein distance and details are introduced in the following subsection}.

{\subsection{Wasserstein distance}}
%\subsection{Robustness in terms of Wasserstein distance}
\begin{definition} (\textbf{Wasserstein distance})
Consider the vectors $x_{1}, x_{2} \in \mathbb{R}^{n}$. Let $\mathcal{P}_{2}(\varsigma_{1},\varsigma_{2})$ denote the collection of all probability measures $\varsigma$ supported on the product space $\mathbb{R}^{2n}$, having finite second moment, with first marginal $\varsigma_{1}$ and second marginal $\varsigma_{2}$. Then the Wasserstein distance of order 2, denoted as ${\mathcal{W}}$, between two probability measures $\varsigma_{1},\varsigma_{2}$, is defined as
\label{Wassdefn}
\begin{align}
&{\mathcal{W}}(\varsigma_{1},\varsigma_{2}) \triangleq \label{W-dist}\\ \nonumber & \left(\displaystyle\inf_{\varsigma\in\mathcal{P}_{2}(\varsigma_{1},\varsigma_{2})}\displaystyle\int_{\mathbb{R}^{2n}} \parallel x_{1}-x_{2}\parallel_{\ell_{2}\left(\mathbb{R}^{n}\right)}^{2} \: d\varsigma(x_{1},x_{2}) \right)
^{\frac{{1}}{2}}.
\end{align}
\end{definition}
\begin{remark}
Intuitively, Wasserstein distance equals the \emph{least amount of work} needed to morph one distributional shape to the other, and can be interpreted as the cost for Monge-Kantorovich optimal transportation plan \cite{villani2003topics}. The particular choice of $\ell_{2}$ norm with order 2 is motivated in \cite{halder2012further}. Further, one can prove (p. 208, \cite{villani2003topics}) that ${\mathcal{W}}$ defines a metric on the manifold of PDFs.
\label{WassRemarkFirst}
\end{remark}

%\begin{remark}
%Given two arbitrary PDFs $\rho_{1}$ and $\rho_{2}$, supported over $\mathbb{R}^{n}$, computing ${\mathcal{W}}$ from (\ref{W-dist}) requires solving a \emph{Hitchcock-Koopmans linear program} (LP), originally formulated in the economics literature \cite{hitchcock1941distribution,koopmans1949optimum,koopmans1951efficient}. In this framework, the PDFs are represented as samples and the sample sizes can be different for each PDF. Another formulation for computing ${\mathcal{W}}$ comes from the computational fluid dynamics community \cite{benamou2000computational}, which involves computation of pressure less potential flow. %As outlined in \cite{halder2012further}, the main computational complexity in solving this LP stems from storage requirement. In particular, if the PDFs $\rho_{1}$ and $\rho_{2}$ are represented as sampled scattered data of cardinality $\nu_{1}$ and $\nu_{2}$, then the sparse storage complexity is $\left(6\nu_{1}\nu_{2} + \left(\nu_{1}+\nu_{2}\right)n + \right.$ $\left.\nu_{1} + \nu_{2}\right)$, and the non-sparse storage complexity is $\left(\nu_{1}+\nu_{2}\right)\left(\nu_{1}\nu_{2} + n + 1\right)$. In other words, the storage complexity for solving the LP at each time, scales \emph{quadratically} with the number of samples (assuming $\nu_{1}, \nu_{2} > n$).
%\label{WassRemarkSecond}
%\end{remark}

Next, we present new results for system stability in terms of ${\mathcal{W}}$ and simplifications in its computation.
 
\begin{proposition}\label{prop:1}
If we fix Dirac distribution as the reference measure, then distributional convergence in Wasserstein metric is \emph{necessary and sufficient} for convergence in m.s. sense.
\label{WConvMeanSqConvProposition}
\end{proposition}
\begin{proof}
Consider a sequence of $n$-dimensional joint PDFs $\{\rho_{j}\left(x\right)\}_{j=1}^{\infty}$, that converges to $\delta\left(x\right)$ in distribution, i.e., $\displaystyle\lim_{j\rightarrow\infty} {\mathcal{W}}\left(\rho_{j}(x), \delta(x)\right) = 0$. From (\ref{W-dist}), we have
\begin{align}
&{\mathcal{W}}^{2}\left(\rho_{j}(x), \delta(x)\right) = \displaystyle\inf_{\varsigma\in\mathcal{P}_{2}(\rho_{j}(x),\delta(x))} \mathbb{E}\left[\parallel X_{j} - 0 \parallel_{\ell_{2}\left(\mathbb{R}^{n}\right)}^{2}\right]\label{IntermedPropWMeanSqConv} \\ \nonumber & = \mathbb{E}\left[\parallel X_{j} \parallel_{\ell_{2}\left(\mathbb{R}^{n}\right)}^{2}\right]
\end{align}
where the random variable $X_{j} \sim \rho_{j}\left(x\right)$, and the last equality follows from the fact that $\mathcal{P}_{2}(\rho_{j}(x),\delta(x)) = \{\rho_{j}(x)\}$ $\forall \: j$, thus obviating the infimum. From (\ref{IntermedPropWMeanSqConv}), $\displaystyle\lim_{j\rightarrow\infty} {\mathcal{W}}\left(\rho_{j}(x), \delta(x)\right) = 0 \Rightarrow \displaystyle\lim_{j\rightarrow\infty} \mathbb{E}\left[\parallel X_{j} \parallel_{\ell_{2}}^{2}\right] = 0$, establishing distributional convergence to $\delta(x) \Rightarrow$ m.s. convergence. Conversely, m.s. convergence $\Rightarrow$ distributional convergence, is well-known \cite{grimmett1992probability} and unlike the other direction, holds for arbitrary reference measure.
\end{proof}

\begin{proposition}(\textbf{$W$ between multivariate Gaussians} \cite{givens1984class})
The Wasserstein distance between two multivariate Gaussians supported on $\mathbb{R}^{n}$, with respective joint PDFs $\mathcal{N}\left(\mu_{1},\Sigma_{1}\right)$ and $\mathcal{N}\left(\mu_{2},\Sigma_{2}\right)$, is given by
\begin{align}
&{\mathcal{W}}\left(\mathcal{N}\left(\mu_{1},\Sigma_{1}\right), \mathcal{N}\left(\mu_{2},\Sigma_{2}\right)\right)=\label{eqn:9} \\ 
&\sqrt{\parallel\mu_{1} - \mu_{2}\parallel_{\ell_{2}\left(\mathbb{R}^{n}\right)}^{2} + \: \text{tr}\left(\Sigma_{1} + \Sigma_{2} - 2 \left[\sqrt{\Sigma_{1}} \Sigma_{2} \sqrt{\Sigma_{1}}\right]^{\frac{1}{2}}\right)}.\nonumber
\end{align}
\end{proposition}

\begin{corollary}\label{cor:1}(\textbf{$W$ between Gaussian and Dirac PDF})
Since we can write $\delta\left(x\right) = \displaystyle\lim_{\mu,\Sigma \rightarrow 0} \mathcal{N}\left(\mu,\Sigma\right)$ (see e.g., p. 160-161, \cite{hassani1999mathematical}), it follows from \eqref{eqn:9} that
\begin{eqnarray}
{\mathcal{W}}\left(\mathcal{N}\left(\mu,\Sigma\right), \delta\left(x\right)\right) = \sqrt{\parallel \mu \parallel_{\ell_{2}\left(\mathbb{R}^{n}\right)}^{2} + \: \text{tr}\left(\Sigma\right)}.\label{GaussianDiracW}
\end{eqnarray}

\label{GaussToDiracCorollary}
\end{corollary}

%\begin{corollary}(\textbf{$W$ between arbitrary and Dirac PDF})
%Given an arbitrary PDF $\rho\left(x\right)$ over $\mathbb{R}^{n}$, it follows from (\ref{W-dist}) that $W^{2}\left(\rho(x),\delta(x)\right) = \displaystyle\int_{\mathbb{R}^{n}} \parallel x \parallel_{\ell_{2}\left(\mathbb{R}^{n}\right)}^{2} \: \rho(x) dx$. If the arbitrary distribution is available as scattered data $\{x_{i}, \varsigma_{i}\}_{i=1}^{\nu}$, then we can algebraically compute the empirical $W$ as $\left(\displaystyle\sum_{i=1}^{\nu} \parallel x_{i} \parallel_{\ell_{2}\left(\mathbb{R}^{n}\right)}^{2} \varsigma_{i}\right)^{\frac{1}{2}}$. 
%\label{ArbitToDiracCorollary}
%\end{corollary}

{\subsection{Performance and Robustness Analysis for SJLSs}}
%\subsection{Robustness Analysis for SJLS}\label{sec3-2:robustness_anal.}
The {performance and} robustness analysis problem for the SJLS is stated as follows: given a SJLS $\left(\pi\left(k\right),A_{\sigma_{k}}\left(x(k)\right),\mathcal{M}\right)$, compute and analyse the performance history, quantified by $W\left(k\right) \triangleq {\mathcal{W}}\left(\rho\left(x(k)\right),\delta(x)\right)$. Comparison of $W(k)$ of uncertain systems with that of a nominal system, quantifies the degradation in system performance due to system uncertainty. 

\subsubsection{Uncertainty propagation in SJLSs}\label{sec:UQLS}
The key difficulty here is the propagation of state PDFs under the stochastic switching and we present a new algorithm for such computations.
\begin{proposition}\label{prop.3}
Given $m$ absolutely continuous random variables $X_{1}, \hdots, X_{m}$, with respective cumulative distribution functions (CDFs) $F_{i}\left(x\right)$, and PDFs $\rho_{i}\left(x\right)$, $\forall i \in \mathcal{M}$. Let $X \triangleq X_{i}$, with probability $\alpha_{i} \in [0,1]$, $\displaystyle\sum_{i=1}^{m} \alpha_{i} = 1$. Then, the CDF and PDF of $X$ are given by $F\left(x\right) = \displaystyle\sum_{i=1}^{m} \alpha_{i} F_{i}\left(x\right)$, and $\rho\left(x\right) = \displaystyle\sum_{i=1}^{m} \alpha_{i} \rho_{i}\left(x\right)$.
\label{RandomVarProposition}
\end{proposition}
\begin{proof}
$F\left(x\right) \triangleq \mathbb{P}\left(X \leq x\right) = \displaystyle\sum_{i=1}^{m} \mathbb{P}\left(X=X_{i}\right) \mathbb{P}\left(X_{i} \leq x\right)$ $= \displaystyle\sum_{i=1}^{m} \alpha_{i} F_{i}\left(x\right)$, where we have used the law of total probability. Since each $X_{i}$ and hence $X$, is absolutely continuous, we have $\rho\left(x\right) = \displaystyle\sum_{i=1}^{m} \alpha_{i} \rho_{i}\left(x\right)$.
\end{proof}
{Note that} any continuous PDF can be approximated by a Gaussian mixture PDF in weak sense \cite{titterington1985statistical,scott1992multivariate}. {Therefore,} we assume the initial PDF for the SJLS to be $m_{0}$ components mixture of Gaussian (MoG), given by $\rho_{0} = \displaystyle\sum_{j_{0}=1}^{m_{0}} \alpha_{j_{0}} \:\mathcal{N}\left(\mu_{j_{0}},\Sigma_{j_{0}}\right)$, $\displaystyle\sum_{j_{0}=1}^{m_{0}} \alpha_{j_{0}} = 1$. Then, we have the following results.

\begin{theorem} (\textbf{A SJLS preserves MoG})
Consider a SJLS $\left(\pi\left(k\right), \{A_{j}\}_{j=1}^{m}, \mathcal{M}\right)$ with initial PDF $\rho_{0} = \displaystyle\sum_{j_{0}=1}^{m_{0}} \alpha_{j_{0}} \:\mathcal{N}\left(\mu_{j_{0}},\Sigma_{j_{0}}\right)$. Then the state PDF at time $k$, denoted by $\rho\left(x(k)\right)$, is given by
{
\begin{align}
\rho\left(x(k)\right) &= \displaystyle\sum_{j_{k}=1}^{m}\displaystyle\sum_{j_{k-1}=1}^{m} \hdots \displaystyle\sum_{j_{1}=1}^{m} \displaystyle\sum_{j_{0}=1}^{m_{0}} \left(\prod_{r=1}^{k} \pi_{j_r}(r)\right)  \nonumber\\ 
&\left.\quad\alpha_{j_{0}}\mathcal{N}\left(\mu_{j_k},\Sigma_{j_k}\right)\right.
\label{dtSJLSstatePDF}
\end{align}
where $\mu_{j_k}=A_{j_k}^*\mu_{j_{0}}$, $\Sigma_{j_k}=A_{j_k}^*\Sigma_{j_{0}}A_{j_k}^{*{\top}}$ and $\displaystyle A_{j_{k}}^*\triangleq \prod_{r=k}^{1}A_{j_r}=A_{j_k}A_{j_{k-1}}\hdots A_{j_{2}}A_{j_{1}}$.}
\label{UncPropThmdtSJLS}
\end{theorem}
\begin{proof}
Starting from $\rho_{0}$ at $k=0$, the modal PDF at time $k=1$, is given by
{
\begin{align}
\rho_{j_1}(x(1)) &= \displaystyle\sum_{j_{0}=1}^{m_{0}}\alpha_{j_{0}} \:\mathcal{N}\left(\mu_{j_1}, \Sigma_{j_1}\right)
\label{ModalPDFatTime1}
\end{align}
where $j_1=1,\cdots ,m$, $\mu_{j_1} = A_{j_1}\mu_{j_0}$, and $\Sigma_{j_1} = A_{j_1}\Sigma_{j_{0}}A_{j_1}^{\top}$,}
which follows from the fact that linear transformation of an MoG is an equal component MoG with linearly transformed component means and congruently transformed component covariances (see Theorem 6 and Corollary 7 in \cite{ali2011convergence}). From Proposition \ref{RandomVarProposition}, it follows that the state PDF at $k=1$, is
{
\begin{align}
\rho(x(1)) = \displaystyle\sum_{j_{1}=1}^{m}\displaystyle\sum_{j_{0}=1}^{m_{0}} \pi_{j_{1}}(1) \alpha_{j_{0}} \:\mathcal{N}\left(\mu_{j_1},\Sigma_{j_1}\right)
\label{dtSJSPDFatTime1}
\end{align}
}
where $\pi_{j_{1}}(1)$ is the occupation probability for mode $j_{1}$ at time $k=1$. Notice that (\ref{dtSJSPDFatTime1}) is an MoG with $m m_{0}$ component Gaussians. Proceeding likewise from this $\rho(x(1))$, we obtain
{
\begin{align}
&\rho_{j_2}(x(2)) = \displaystyle\sum_{j_{1}=1}^{m}\displaystyle\sum_{j_{0}=1}^{m_{0}}\pi_{j_{1}}(1) \alpha_{j_{0}} \:\mathcal{N}\big(\mu_{j_2},\Sigma_{j_2}\big)\\ 
&\text{where }j_2=1,\hdots,m,\: \mu_{j_2}=(A_{j_2}A_{j_{1}})\mu_{j_{0}},\nonumber\\
&\qquad\qquad\qquad\qquad\quad\:\Sigma_{j_2}=(A_{j_2}A_{j_{1}})\Sigma_{j_{0}}(A_{j_2}A_{j_{1}})^{\top},\nonumber\\
&\rho(x(2))=\displaystyle\sum_{j_{2}=1}^{m}\displaystyle\sum_{j_{1}=1}^{m}\displaystyle\sum_{j_{0}=1}^{m_{0}}\pi_{j_{2}}(2)\pi_{j_{1}}(1) \alpha_{j_{0}} \:\mathcal{N}\big(\mu_{j_2},\Sigma_{j_2}\big).
%&\text{where }\mu_{j_2}=(A_{j_{2}}A_{j_{1}})\mu_{j_{0}}, \Sigma_{j_2}=(A_{j_{2}}A_{j_{1}})\Sigma_{j_{0}}(A_{j_{2}}A_{j_{1}})^{\top}.\nonumber
\label{ModalStatePDFatTime2}
\end{align}
}
Continuing with this recursion till time $k$, we arrive at (\ref{dtSJLSstatePDF}), which is an MoG with $m^{k}m_{0}$ components. We comment that the expression simplifies for $m_{0} = 1$, i.e. when the initial PDF is Gaussian.
\end{proof}
\begin{remark}
(\textbf{Computational complexity}) Given an initial MoG and a SJLS, from Theorem \ref{UncPropThmdtSJLS}, one can in principle compute the state PDF at any {finite} time, in closed form {(i.e., an analytical form with a finite number of well-defined functions)}. However, since the number of component Gaussians grow{s} exponentially in time, the computational complexity in evaluating (\ref{dtSJLSstatePDF}), grows exponentially, and hence the computation becomes intractable. In the following, we show that the Wasserstein based performance analysis can still be performed in closed form while keeping the computational complexity constant in time.
\label{ComplexityRemark}
\end{remark}

\subsubsection{$Wasserstein$ computation in SJLSs}\label{sec:3.2.2}
For a SJLS, there are no known results to represent the $W$ distance in closed form. The main computational issue is that even with Gaussian initial PDF, the instantaneous state PDF remains no longer Gaussian but rather MoG, as shown in Theorem \ref{UncPropThmdtSJLS}. This brings forth concerns for the exponential growth of computational complexity to obtain $\rho(x(k))$.
{To address these concerns, we firstly introduce a following theorem that enables the Wasserstein computation in an analytical form. Then, we further show that the exponential growth can be obviated by the proposed algorihm.}
 %However, we introduce new results that show the existence of the analytical solution in the exact closed form to cope with above concerns for the robustness analysis of SJLS. %We need the following lemma, proposition, and theorems to address these concerns.

{
\begin{theorem} (\textbf{$W$ for an $m$-mode SJLS with Dirac reference PDF})
At any given time $k$, let the state PDF for a SJLS be $\rho(x) = \displaystyle\sum_{j=1}^{m} \alpha_{j} \rho_{j}(x)$, $x \in \mathbb{R}^{n}$ where $\rho_{j}(x)$, $\alpha_{j}$, and $m$ are the instantaneous modal PDF, time-varying occupation probability of mode $j$, and the number of individual mixture components, respectively. If we define $W \triangleq {\mathcal{W}}\big(\rho\left(x\right), \delta(x)\big)$, and $W_{j} \triangleq {\mathcal{W}}\big(\rho_{j}\left(x\right), \delta(x)\big)$, then
\begin{equation}
W = \left(\displaystyle\sum_{j=1}^{m} \alpha_{j} W_{j}^{2}\right)^{1/2}.
\label{WrelationSJS}
\end{equation}
\label{WassThmSJSwithDirac}
\end{theorem}
}
\begin{figure*}
\begin{center}
\includegraphics[width=\textwidth]{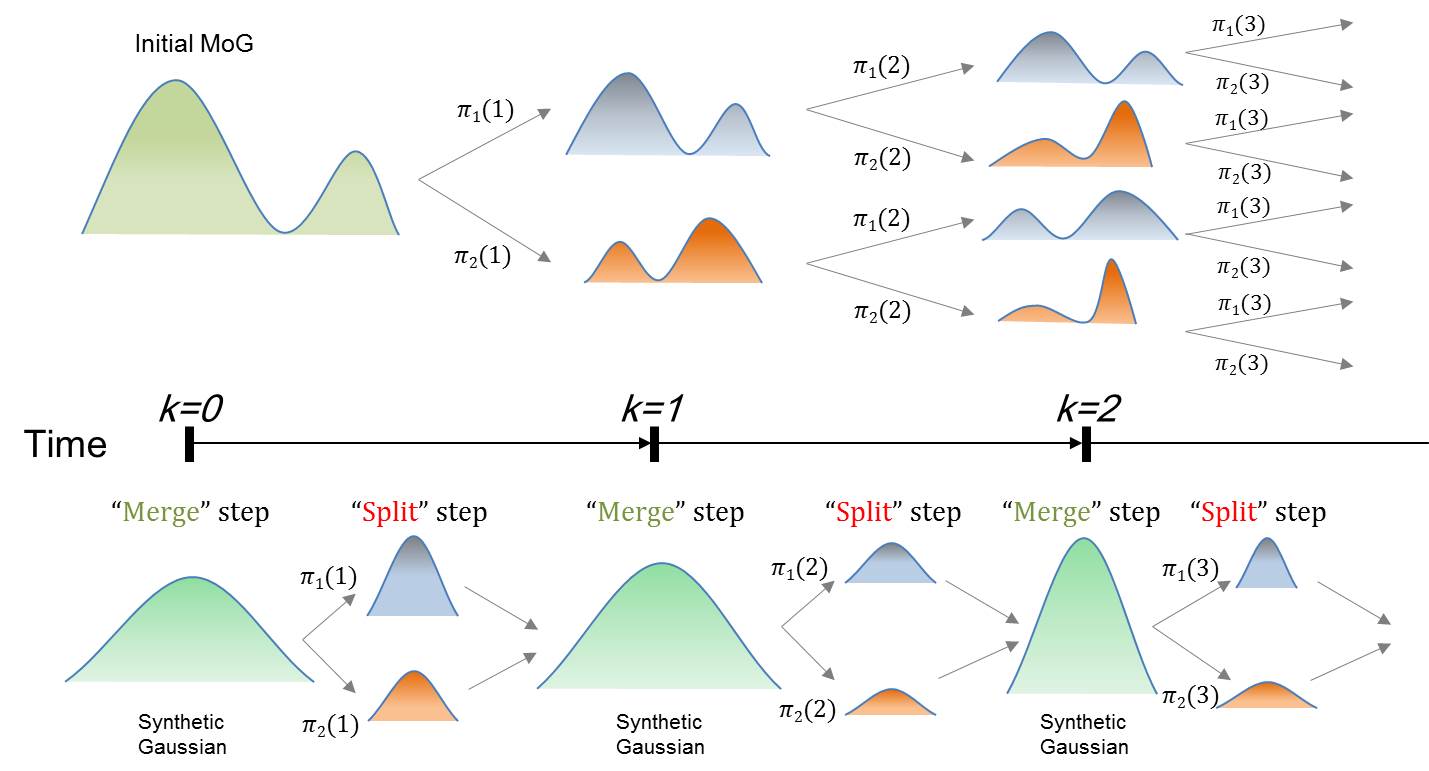}
\caption{Schematic of PDFs propagation for SJLS. Initially, an MoG PDF was given; Upper one shows the exponential growth of MoG components; Bottom one shows ``Split-and-Merge'' algorithm and the number of Gaussian components remains constatnt, which is $m$ modes at most. In this figure, $m=2$.}\label{fig_tree}
\end{center}
\end{figure*}

{
\begin{proof}
From (\ref{W-dist}) and Proposition \ref{prop.3}, we have
\begin{eqnarray}
\qquad\qquad \nonumber W^{2} &=& \displaystyle\int_{\mathbb{R}^{n}} \parallel x \parallel_{\ell_{2}\left(\mathbb{R}^{n}\right)}^{2} \rho(x) dx\\
&=& \displaystyle\int_{\mathbb{R}^{n}} \parallel x \parallel_{\ell_{2}\left(\mathbb{R}^{n}\right)}^{2} \displaystyle\sum_{j=1}^{m} \alpha_{j} \rho_{j}(x) dx \nonumber \\
&=& \displaystyle\sum_{j=1}^{m} \alpha_{j} \displaystyle\int_{\mathbb{R}^{n}} \parallel x \parallel_{\ell_{2}\left(\mathbb{R}^{n}\right)}^{2} \rho_{j}(x) dx \nonumber \\ 
&=& \displaystyle\sum_{j=1}^{m} \alpha_{j} W_{j}^{2}.
\end{eqnarray}
\begin{eqnarray}
\qquad\qquad\qquad\Rightarrow W = \left(\displaystyle\sum_{j=1}^{m} \alpha_{j} W_{j}^{2}\right)^{1/2}.\label{WassAsConvexSumOfComponentWass} 
\end{eqnarray}
%To prove the inequality relation, we invoke Proposition \ref{WassPropositionGaussNonGaussInequality} and recall that $\delta\left(x\right) = \displaystyle\lim_{\mu,\Sigma \rightarrow 0} \mathcal{N}\left(\mu,\Sigma\right)$ (Corollary \ref{GaussToDiracCorollary}). This results in $\widehat{W} \leq W$, and we conclude the proof.
\end{proof}
}

{
Theorem \ref{WassThmSJSwithDirac} provides an analytical solution to compute the performance and the robustness of the SJLS in terms of Wasserstein distance. However, expression in \eqref{WrelationSJS} still includes the component-wise $W$ computation, and hence the computation becomes intractable shortly due to the exponential growth of Gaussian components in the state PDF $\rho(x)$. In order to cope with this problem, we introduce a \textbf{``Split-and-Merge''} algorithm as follows.
}

{
\textbf{1) Merge Step:}\\
For a given MoG $\rho(x)$ at any time $k$, we can compute the mean $\widehat{\mu}$ and covariance $\widehat{\Sigma}$ of an MoG by the following lemma.}

\begin{lemma} (\textbf{Mean and covariance of a mixture PDF})\label{lemma:1}
Consider any mixture PDF $\rho(x) = \displaystyle\sum_{j=1}^{m} \alpha_{j} \rho_{j}(x)$, with component mean-covariance pairs $\left(\mu_{j},\Sigma_{j}\right)$, $j=1,\hdots,m$. Then, the mean-covariance pair $\left(\widehat{\mu},\widehat{\Sigma}\right)$ for the mixture PDF $\rho(x)$, is given by
\begin{align}
\widehat{\mu} = \sum_{j=1}^{m}\alpha_{j} \mu_{j}, \,\widehat{\Sigma} = \sum_{j=1}^{m} \alpha_{j}\left(\Sigma_{j} + \left(\mu_{j}-\widehat{\mu}\right)\left(\mu_{j}-\widehat{\mu}\right)^{\top}\right).
\label{MeanCovHat}
\end{align}
\label{MeanCovMixPDF}
\end{lemma}
\begin{proof}
We have $\widehat{\mu} \triangleq \displaystyle\int_{\mathbb{R}^{n}} x \rho(x) dx = \displaystyle\int_{\mathbb{R}^{n}} x \displaystyle\sum_{j=1}^{m} \alpha_{j}$  $\rho_{j}(x)dx =\displaystyle\sum_{j=1}^{m} \alpha_{j} \displaystyle\int_{\mathbb{R}^{n}} x \rho_{j}(x) dx = \displaystyle\sum_{j=1}^{m}\alpha_{j} \mu_{j}$.

On the other hand, $\widehat{\Sigma} \triangleq \mathbb{E}\left[\left(x-\widehat{\mu}\right) \left(x-\widehat{\mu}\right)^{\top}\right] = \mathbb{E}\left[x x^{\top}\right] - \widehat{\mu}\widehat{\mu}^{\top} = \displaystyle\int_{\mathbb{R}^{n}} x x^{\top} \displaystyle\sum_{j=1}^{m} \alpha_{j} \rho_{j}(x) dx - \widehat{\mu}\widehat{\mu}^{\top} = \displaystyle\sum_{j=1}^{m} \alpha_{j} \displaystyle\int_{\mathbb{R}^{n}} \left(x - \widehat{\mu} + \widehat{\mu}\right) \left(x - \widehat{\mu} + \widehat{\mu}\right)^{\top} \rho_{j}\left(x\right) dx - \widehat{\mu}\widehat{\mu}^{\top} = \displaystyle\sum_{j=1}^{m}\alpha_{j}\left(\Sigma_{j}+\left(\mu_{j}-\widehat{\mu}\right)\left(\mu_{j}-\widehat{\mu}\right)^{\top}\right)$.
\end{proof}
{
Lemma \ref{MeanCovMixPDF} proves that for any mixture PDF, we can compute the mean $\widehat{\mu}$ and covariance $\widehat{\Sigma}$. From the computed $\widehat{\mu}(k)$ and $\widehat{\Sigma}(k)$ at time $k$, we construct a synthetic Gaussian $\mathcal{N}(\widehat{\mu}(k),\widehat{\Sigma}(k))$ to merge the state PDF of an MoG form into a single Gaussian PDF.}

{
\textbf{2) Split Step:}\\
Once the synthetic Gaussian $\mathcal{N}(\widehat{\mu}(k),\widehat{\Sigma}(k))$ is obtained at time $k$ from ``Merge step", we proceed the propagation of the modal PDF for the next time step along mode dynamics $\{A_j\}_{j=1}^{m}$. Consequently, we have $m$ numbers of Gaussian components $\mathcal{N}(A_j\widehat{\mu}(k), A_j\widehat{\Sigma}(k)A_j^{\top})$, $j=1,2,\hdots,m$ at time $k+1$.}
%According to Proposition \ref{prop.3} with switching probability $\pi(k+1)$, the state PDF at time $k+1$ forms a new MoG expressed by $\widehat{\rho}(x(k+1)) = \sum_{j=1}^{m}\pi_j(k+1)\mathcal{N}(A_j\widehat{\mu}(k), A_j\widehat{\Sigma}(k){A_j}^{\top})$.
%Since the propagation started from a synthetic Gaussian $\mathcal{N}(\widehat{\mu},\widehat{\Sigma})$, we have $m$ number of Gaussian components at time $k+1$, which forms an MoG.}

{
Repeating ``Split-and-Merge" algorithm at every time step as depicted by Fig. \ref{fig_tree}, linear modal dynamics results in $m$ modal Gaussian PDFs (``Split step"). Then, instead of computing the non-Gaussian SJLS state PDF in an MoG form, one would construct a synthetic Gaussian $\mathcal{N}(\widehat{\mu},\widehat{\Sigma})$ (``Merge step") and repeat thereafter. %Thus, at any time $k$, we only have $m$ mean vectors and covariance matrices to work with.

Although the ``Split-and-Merge" algorithm obviate the need to compute the state PDF $\rho(x)$ where Gaussian components grow exponentially, the synthetic Gaussian PDF $\mathcal{N}(\widehat{\mu},\widehat{\Sigma})$ does not imply that it can replace $\rho(x)$. Since $\rho(x)$ expressed in an MoG form have higher moments other than first and second, the distance between $\rho(x)$ and $\delta(x)$ may differ from that between $\mathcal{N}(\widehat{\mu},\widehat{\Sigma})$ and $\delta(x)$. However, most importantly, we address that ${\mathcal{W}}(\rho(x), \delta(x))$ and ${\mathcal{W}}(\mathcal{N}(\widehat{\mu},\widehat{\Sigma}),\delta(x))$ are equidistant at any time $k$ by the following theorem.
}

%\begin{proposition}(\textbf{A generic lower bound for $W$} (p. 11, \cite{rachev1998mass}))
%Consider two arbitrary PDFs $\rho_{1}$ and $\rho_{2}$, with respective mean-covariance pairs $\left(\mu_{1},\Sigma_{1}\right)$ and $\left(\mu_{2},\Sigma_{2}\right)$. Then, we have
%\begin{equation}
%W\left(\mathcal{N}\left(\mu_{1},\Sigma_{1}\right), \mathcal{N}\left(\mu_{2},\Sigma_{2}\right)\right) \leq W\left(\rho_{1}, \rho_{2}\right).
%\end{equation}
%\label{WassPropositionGaussNonGaussInequality}
%\end{proposition}

{
\begin{theorem} (\textbf{Equidistance between $W$ and $\widehat{W}$})\label{theorem_What=W}
At any given time $k$, let the state PDF for an $m$-mode SJLS $\rho(x(k))$, be of the form (\ref{dtSJLSstatePDF}), which we rewrite as $\rho\left(x(k)\right) = \displaystyle\sum_{j_{k}=1}^{m}\displaystyle\sum_{j_{0}=1}^{m_{0}} \alpha_{j_{0}}\beta_{j_{k}} \mathcal{N}\left(\mu_{j_k},\Sigma_{j_k}\right)$, where $\beta_{j_{k}}\triangleq\displaystyle\sum_{j_{k-1}=1}^{m}\hdots\displaystyle\sum_{j_{1}=1}^{m} \left(\prod_{r=1}^{k} \pi_{j_r}(r)\right)$, $\mu_{j_k} = A_{j_{k}}^{*} \mu_{j_{0}}$, $\Sigma_{j_k}=A_{j_{k}}^{*} \Sigma_{j_{0}} A_{j_{k}}^{*^{\top}}$, and $\displaystyle A_{j_k}^*\triangleq \prod_{r=k}^{1}A_{j_r}$. Let the instantaneous mean and covariance of the mixture PDF $\rho(x(k))$ be $\widehat{\mu}(k)$ and $\widehat{\Sigma}(k)$, respectively.  Then, we have
\begin{align}
&\widehat{W}(k) = W(k) = \left(\displaystyle\sum_{j_{k}=1}^{m}\displaystyle\sum_{j_{0}=1}^{m_{0}} \alpha_{j_{0}} \beta_{j_{k}} W_{j_{k}}^{2}(k)\right)^{1/2}, \forall k\in\mathbb{N}_0
\label{WrelationSJLS}\\
&\text{where }\nonumber\\
&\qquad \widehat{W}(k) \triangleq {\mathcal{W}}\left(\mathcal{N}\left(\widehat{\mu}(k),\widehat{\Sigma}(k)\right), \delta(x)\right),\nonumber\\ 
&\qquad W(k) \triangleq {\mathcal{W}}\left(\rho\left(x(k)\right), \delta(x)\right),\nonumber\\
&\qquad W_{j_{k}}(k) \triangleq {\mathcal{W}}\left(\mathcal{N}\left(\mu_{j_k},\Sigma_{j_k}\right), \delta(x)\right),\nonumber\\
&\qquad \mu_{j_k} = A_{j_{k}}^{*} \mu_{j_{0}},\: \Sigma_{j_k} = A_{j_{k}}^{*}\Sigma_{j_0}A_{j_{k}}^{*^{\top}},\quad \forall k\geq 1.\nonumber
\end{align}
\label{WassThmSJLSwithDirac}
\end{theorem}

\begin{proof}
The rightmost equality in (\ref{WrelationSJLS}), follows directly from Theorem \ref{WassThmSJSwithDirac}. Thus, it suffices to prove that $\widehat{W}(k)=\left(\sum_{j_{k}=1}^{m}\sum_{j_0=1}^{m_0}\alpha_{j_0}\beta_{j_{k}} W_{j_{k}}^{2}(k)\right)^{1/2}$.

At time $k=0$, the mean and covariance pair $(\widehat{\mu}_0, \widehat{\Sigma}_0)$ of an initial MoG can be computed by $(\widehat{\mu}_0, \widehat{\Sigma}_0) = \big(\sum_{j_0=1}^{m_0}\alpha_{j_0}\mu_{j_0}, \sum_{j_0=1}^{m_0}(\Sigma_{j_0}+(\mu_{j_0}-\widehat{\mu}_0)(\mu_{j_0}-\widehat{\mu}_0)^{\top})\big)$ from Lemma \ref{lemma:1}. If we construct a synthetic Gaussian $\mathcal{N}(\widehat{\mu}_0, \widehat{\Sigma}_0)$, Wasserstein distance $\widehat{W}$ at time $k=0$ can be computed by \eqref{GaussianDiracW} as follows.
\begin{align}
&\widehat{W}^2(0) \overset{\eqref{GaussianDiracW}}{=} \parallel \widehat{\mu}_0\parallel_{\ell_{2}\left(\mathbb{R}^{n}\right)}^{2} + \text{tr}(\widehat{\Sigma}_0) \overset{(\ref{MeanCovHat})}{=} \widehat{\mu}_0^{\top}\widehat{\mu}_0 + \nonumber \\ 
&\text{tr}\left(\displaystyle\sum_{j_{0}=1}^{m_0} \alpha_{j_{0}}\Big(\Sigma_{j_{0}} + (\mu_{j_{0}}-\widehat{\mu}_0)(\mu_{j_{0}}-\widehat{\mu}_0)^{\top}\Big)\right).\label{eqn:18}
\end{align}
Since $\text{tr}(\cdot)$ is a linear operator, we can expand \eqref{eqn:18} as
\begin{align}
&\widehat{W}^{2}(0) = \widehat{\mu}_0^{\top}\widehat{\mu}_0
+ \: \displaystyle\sum_{j_{0}=1}^{m_0}\alpha_{j_0}\text{tr}\left(\Sigma_{j_{0}}\right) + \: \text{tr}\left(\displaystyle\sum_{j_{0}=1}^{m_0} \alpha_{j_0}\mu_{j_{0}}\mu_{j_{0}}^{\top}\right) \nonumber \\ 
& - \text{tr}\left(\left(\displaystyle\sum_{j_{0}=1}^{m_0} \alpha_{j_0}\mu_{j_{0}}\right)\widehat{\mu}_0^{\top}\right) - \text{tr}\left(\widehat{\mu}_0\left(\displaystyle\sum_{j_{0}=1}^{m_0} \alpha_{j_0}\mu_{j_{0}}\right)^{\top}\right) \nonumber \\ 
& + \text{tr}\left(\widehat{\mu}_0\widehat{\mu}_0^{\top}\right).\label{eqn:19}
\end{align}
Recalling that $\widehat{\mu}_0 = \sum_{j_0=1}^{m_0}\alpha_{j_0}\mu_{j_0}$ and $\widehat{\mu}_0^{\top}\widehat{\mu}_0 = \:\text{tr}\left(\widehat{\mu}_0^{\top}\widehat{\mu}_0\right) = \:\text{tr}\left(\widehat{\mu}_0\widehat{\mu}_0^{\top}\right)$, the first, fourth, fifth and sixth term in the right-hand-side of \eqref{eqn:19} cancel out, resulting in
\begin{eqnarray}
\qquad\widehat{W}^{2}(0) &=& \displaystyle\sum_{j_{0}=1}^{m_0} \alpha_{j_{0}} \: \text{tr}\left(\mu_{j_{0}}\mu_{j_{0}}^{\top}\right) + \: \displaystyle\sum_{j_{0}=1}^{m_0}\alpha_{j_{0}}\text{tr}\left(\Sigma_{j_{0}}\right)\nonumber\\
\qquad&=& \displaystyle\sum_{j_{0}=1}^{m_0}\alpha_{j_{0}} \left(\parallel \mu_{j_{0}} \parallel_{\ell_{2}\left(\mathbb{R}^{n}\right)}^{2} + \: \text{tr}\left(\Sigma_{j_{0}}\right)\right) \nonumber\\
&=& \displaystyle\sum_{j_{0}=1}^{m_0}\alpha_{j_{0}} {\mathcal{W}}^2\Big(\mathcal{N}(\mu_{j_0},\Sigma_{j_0}), \delta(x)\Big)\nonumber\\
&=& \displaystyle\sum_{j_{0}=1}^{m_0}\alpha_{j_{0}} W_{j_0}^{2}(0) \overset{\eqref{WrelationSJS}}{=} W^2(0).\label{eqn:20}
\end{eqnarray}
Hence, $\widehat{W}(0)$ is equidistant with $W(0)$.
%\begin{figure}
%\begin{center}
%\includegraphics[width=0.475\textwidth]{./Figures/CombiningTex_scissored.pdf}
%\caption{Illustration of Theorem \ref{WassThmSJLSwithDirac}, showing that given MoG PDF $\rho$, we can construct Gaussian $\widehat{\rho}$ such that $W = \widehat{W}$, where $W\triangleq W\left(\rho,\delta\right)$, and $\widehat{W}\triangleq W\left(\widehat{\rho},\delta\right)$.}
%\label{IllustrateMOGtheorem}
%\end{center}
%\end{figure}

At time $k=1$, we propagate the modal PDFs from a synthetic Gaussian $\mathcal{N}(\widehat{\mu}_{0},\widehat{\Sigma}_{0})$, which results in $m$ modal Gaussians $\mathcal{N}(A_{j_1}\widehat{\mu}_{0}, A_{j_1}\widehat{\Sigma}_{0}A_{j_1}^{\top})$, $j_1 = 1,2,\hdots,m$ during ``Split step", followed by ``Merge step" to obtain a new synthetic Gaussian $\mathcal{N}(\widehat{\mu}_{1},\widehat{\Sigma}_{1})$, where $\widehat{\mu}_{1}=\sum_{j_1=1}^{m}\pi_{j_1}(1)A_{j_1}\widehat{\mu}_{0}$ and $\widehat{\Sigma}_{1}=\sum_{j_1=1}^{m}\pi_{j_1}(1)\bigg(A_{j_1}\widehat{\Sigma}_{0}A_{j_1}^{\top} + (A_{j_1}\widehat{\mu}_{0}-\widehat{\mu}_{1})(A_{j_1}\widehat{\mu}_{0}-\widehat{\mu}_{1})^{\top}\bigg)$ from Lemma \ref{MeanCovMixPDF}.
Then, $\widehat{W}(1)$ can be computed by
\begin{align}
&\widehat{W}^{2}(1) \overset{\eqref{GaussianDiracW}}{=}\parallel \widehat{\mu}_{1}\parallel_{\ell_{2}\left(\mathbb{R}^{n}\right)}^{2} + \text{tr}\big(\widehat{\Sigma}_{1}\big) 
= \widehat{\mu}_{1}^{\top}\widehat{\mu}_{1} + \tr\Bigg(\sum_{j_1=1}^{m}\nonumber\\
&\pi_{j_1}(1)\Big( A_{j_1}\widehat{\Sigma}_{0}A_{j_1}^{\top} + \big(A_{j_1}\widehat{\mu}_{0}-\widehat{\mu}_{1})(A_{j_1}\widehat{\mu}_{0}-\widehat{\mu}_{1}\big)^{\top}\Big)\Bigg).\nonumber\\
\end{align}
By exactly the same procedure in \eqref{eqn:19}, and the term cancellation, we arrive at
\begin{align}
&\widehat{W}^{2}(1) = \sum_{j_1=1}^{m}\pi_{j_1}(1)\bigg(\tr\Big(A_{j_1}\widehat{\mu}_{0}\widehat{\mu}_{0}^{\top}A_{j_1}^{\top} + A_{j_1}\widehat{\Sigma}_{0}A_{j_1}^{\top}\Big)\bigg)\nonumber\\
&\overset{\eqref{MeanCovHat}}{=} \sum_{j_1=1}^{m}\pi_{j_1}(1)\Bigg(\tr\bigg(A_{j_1}\Big(\sum_{j_0=1}^{m_0}\alpha_{j_0}\big(\mu_{j_0}\mu_{j_0}^{\top}+\Sigma_{j_0}\big)\Big)A_{j_1}^{\top} \bigg)\Bigg)\nonumber\\
&= \sum_{j_1=1}^{m}\sum_{j_0=1}^{m_0}\pi_{j_1}(1)\alpha_{j_0}\bigg(\parallel \mu_{j_1} \parallel_{\ell_{2}\left(\mathbb{R}^{n}\right)}^{2} + \text{tr}\big(\Sigma_{j_1}\big)\bigg)\nonumber\\
&= \sum_{j_1=1}^{m}\sum_{j_0=1}^{m_0}\pi_{j_1}(1)\alpha_{j_0}W^2_{j_1}(1) \overset{\eqref{WrelationSJS}}{=} W^2(1)
\end{align}
where $\mu_{j_1}=A_{j_1}\mu_{j_0}$ and $\Sigma_{j_1}=A_{j_1}\Sigma_{j_0}A_{j_1}^{\top}$.

Continuing in this manner, finally we obtain a following result for any time $k$.
\begin{eqnarray}
\widehat{W}^{2}(k) &=& \sum_{j_k=1}^{m}\cdots\sum_{j_1=1}^{m}\sum_{j_0=1}^{m_0}\Bigg(\prod_{r=1}^{k}\pi_{j_r}(r)\Bigg)\alpha_{j_0}\nonumber\\
&& \quad\qquad\qquad\qquad \bigg(\parallel \mu_{j_k} \parallel_{\ell_{2}\left(\mathbb{R}^{n}\right)}^{2} + \text{tr}\big(\Sigma_{j_k}\big)\bigg)\nonumber\\
&=& \sum_{j_k=1}^{m}\cdots\sum_{j_1=1}^{m}\sum_{j_0=1}^{m_0}\Bigg(\prod_{r=1}^{k}\pi_{j_r}(r)\Bigg)\alpha_{j_0}W^2_{j_k}(k)\nonumber \\
&\overset{\eqref{WrelationSJS}}{=}& W^2(k)
\end{eqnarray}
where $\mu_{j_k} = A_{j_k}A_{j_{k-1}}\cdots A_{j_1}\mu_{j_0}=A_{j_k}^{*}\mu_{j_0}$,\\ $\Sigma_{j_k} = $ $\big(A_{j_k}A_{j_{k-1}}\cdots A_{j_1}\big)\Sigma_{j_0}\big(A_{j_k}A_{j_{k-1}}\cdots A_{j_1}\big)^{\top}=A_{j_k}^{*}\Sigma_{j_0}A_{j_k}^{*\top}$.
\end{proof}

According to Theorem \ref{theorem_What=W}, it is unnecessary to propagate the state PDF $\rho(x)$ and to compute $W$, which is intractable due to the exponential growth of Gaussian components. Instead, we can analyse the performance of the SJLS through $\widehat{W}$, since $\widehat{W}$ is equidistant with $W$ at all time $k$. 
%The practical utilization of the ``Split-and-Merge" algorithm with $\widehat{W}$ computation can be stated as follows. From a synthetic Gaussian $\mathcal{N}(\widehat{\mu}(k), \widehat{\Sigma}(k))$ obtained by ``Merge step" at time $k$, we propagate modal PDFs by mode dynamics $\{A_j\}_{j=1}^{m}$ during ``Split step". At time $k+1$, we apply the ``Merge step" and obtain a synthetic Gaussian $\mathcal{N}(\widehat{\mu}(k+1), \widehat{\Sigma}(k+1))$, followed by $\widehat{W}(k+1)$ computation with \eqref{GaussianDiracW}. 
The major advantages of the ``Split-and-Merge" algorithm with $\widehat{W}$ computation for the performance and the robustness analysis can be summarized in the following sense. $\widehat{W}$ computation using \eqref{GaussianDiracW} provides an analytical solution, which is computationally concise and efficient enough. In addition, at any time step, we only have $m$ mean vectors and covariance matrices to work with, and hence the scalability problem with an exponential growth can be avoided.
}

{
\begin{remark}\textbf{(Applicability of performance and robustness measure to general SJLSs)}
Since the switching probability $\pi(k)$ is an independent variable with regard to $\widehat{W}(k)$ as described in Theorem \ref{theorem_What=W}, we can compute $\widehat{W}(k)$ for any SJLSs regardless of the updating rule for $\pi(k)$. Once $\pi(k)$ is computed at time $k$ by governing recursion equation (i.e., i.i.d., Markov, or semi-Markov jump process, etc.), the performance and the robustness for SJLSs are measured by $\widehat{W}(k)$. As a consequence, the proposed method for the performance and robustness measure can be applied to any SJLSs.
\end{remark}
}
%-------------------------------------------------------------------------------------------------------------

\section{Numerical Example}
%The proposed methods for the robustness analysis are applicable to any stochastic jump linear systems, not necessarily Markovian jumps. However, since asynchronous behavior such as communication delays or packet losses are being widely represented by Markovian process, we specify a system with random communication delays, which has also Markovian process.

Consider the inverted pendulum on cart in Fig. \ref{fig:inv_pendulum} with parameters described in Table \ref{table_inverted_pendulum}. Originally, this example was introduced in \cite{xiao2000control} with single communication delay term $\tau_k$ between sensor and controller. 

\begin{table}[h]
\begin{center}
\caption{Nomenclature for Inverted Pendulum Dynamics.}\label{table_inverted_pendulum}
\label{table_quadrotor_nomenclature}
  \begin{tabular}{|c|c|c|c|}\hline
  Symbol & definition & Symbol & definition\\\hline
	$m_1$ & cart mass & $m_2$ & pendulum mass\\
	$L$ & pendulum length & $x$ & cart position\\
	$\theta$ & pendulum angle & $u$ & input force
  \\\hline
  \end{tabular}
\end{center}
\end{table}

The system states are $x_1 = x$, $x_2 = \dot{x}$, $x_3 = \theta$, and $x_4 = \dot{\theta}$. We assume that $m_1 = 1$kg, $m_2 = 0.5$kg, $L=1$m with friction-free floor.
Later, this example was further exploited by \cite{zhang2005new} with two random delays $\tau_k$ and $d_k$ which are sensor-to-controller and controller-to-actuator delays, respectively. The sets of mode are $\mathcal{M}(\tau_k) = \{0,1,2\}$ and $\mathcal{M}(d_k)=\{0,1\}$.
When the control action is taken at time $k$, the controller-to-actuator delay $d_k$ is unknown, but $\tau_k$ and $d_{k-1}$ are found. Accordingly, controller gain $F$ is dependent on $\tau_k$ and $d_{k-1}$. Hence, the linearized closed-loop system model with sampling time $T_s = 0.1$ is denoted by
\begin{align*}
x(k+1) = Ax(k)+BF(\tau_k,d_{k-1})x(k-\tau_k-d_k)
\end{align*}
where
\begin{align*}
A = \begin{bmatrix}
1 & 0.1 & -0.0166 & -0.0005\\
0 & 1 & -0.3374 & -0.0166\\
0 & 0 & 1.0996 & 0.1033\\
0 & 0 & 2.0247 & 1.0996
\end{bmatrix},\quad
B = \begin{bmatrix}
0.0045\\0.0896\\-0.0068\\-0.1377
\end{bmatrix}
\end{align*}
\begin{figure}
\begin{center}
\includegraphics[scale=0.4]{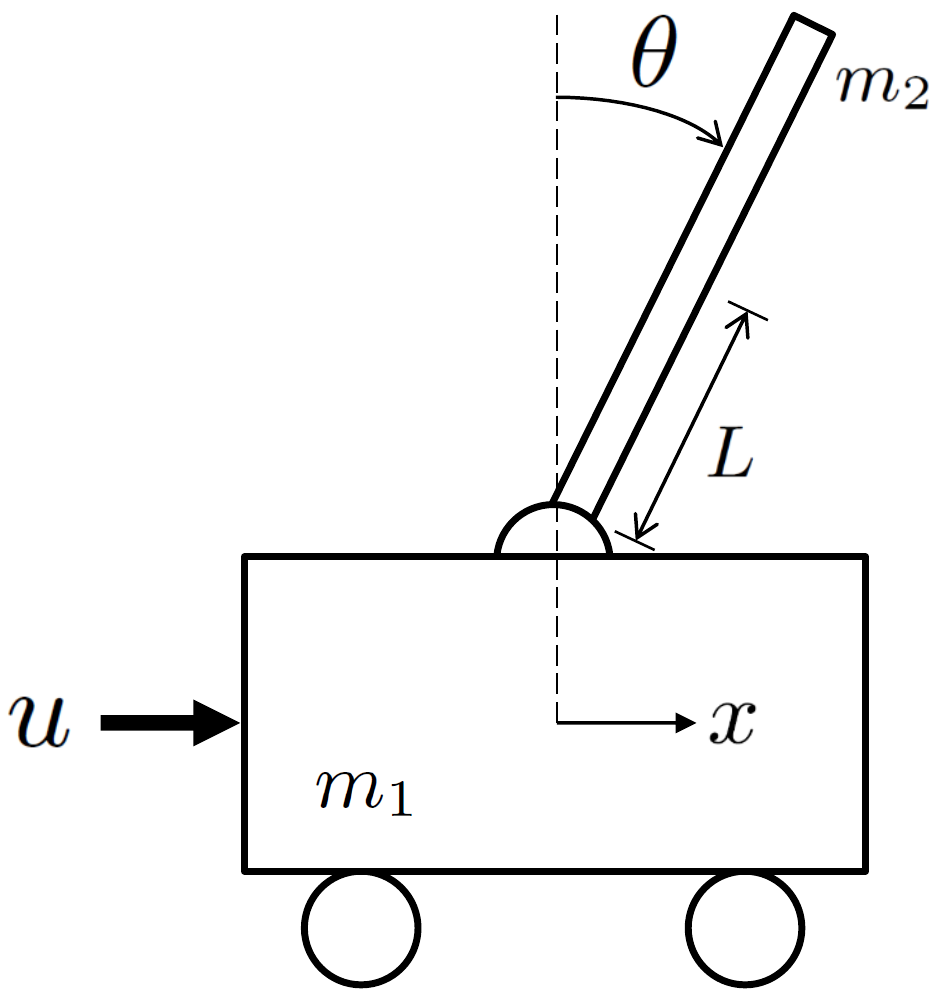}
\caption{Inverted Pendulum on Cart.}\label{fig:inv_pendulum}
\end{center}
\end{figure} 

with the controller gain $F$'s given in \cite{zhang2005new}:
\begin{align*}
&F(0,0) = \begin{bmatrix}
0.1690 & 0.8824 & 19.5824 & 4.3966
\end{bmatrix}\\
&F(0,1) = \begin{bmatrix}
0.5625 & 0.6259 & 24.8814 & 5.1886
\end{bmatrix}\\
&F(1,0) = \begin{bmatrix}
-0.3076 & 0.9370 & 12.0069 & 5.9910
\end{bmatrix}\\
&F(1,1) = \begin{bmatrix}
-0.0097 & 0.7109 & 15.2518 & 7.3154
\end{bmatrix}\\
&F(2,0) = \begin{bmatrix}
-0.3212 & 1.0528 & 11.9330 & 6.3809
\end{bmatrix}\\
&F(2,1) = \begin{bmatrix}
0.0427 & 0.8640 & 16.0874 & 7.8361
\end{bmatrix}.\\
\end{align*}
Therefore, this system has total $6$ numbers of closed-loop dynamics $A_{\sigma_k}$ with $\sigma_k\in\{1,2,\hdots,6\}$.

{
1) Markovian Communication Delays:\\
We denote the transition probability of sensor-to-controller and controller-to-actuator delays as $\lambda_{ij}$ and $\omega_{rs}$, respectively. Then, $\lambda_{ij}$ and $\omega_{rs}$ are defined by
}
\begin{align*}
\lambda_{ij} = \mathbb{P}(\tau_{k+1}=j|\tau_{k}=i),\:\omega_{rs} = \mathbb{P}(\omega_{k+1}=s|\omega_{k}=r)
\end{align*}
where $\lambda_{ij},\omega_{rs}\geq 0$ and $\sum_{j=0}^{2}\lambda_{ij}=1$, $\sum_{s=0}^{1}\omega_{rs}=1$. 
{Given individual Markov transition probability matrices}
\begin{align*}
\Lambda = \begin{bmatrix}
0.5 & 0.5 & 0\\
0.3 & 0.6 & 0.1\\
0.3 & 0.6 & 0.1
\end{bmatrix}, \quad
\Omega = \begin{bmatrix}
0.2 & 0.8\\
0.5 & 0.5
\end{bmatrix}
\end{align*}
{
corresponding to $\lambda_{ij}$ and $\omega_{rs}$, the Markov transition probability matrix $P$ for 6 modes MJLS is obtained from $P=\Lambda\otimes \Omega$ as in \cite{xiao2000control}. The switching probability distribution $\pi(k)$ is updated by the linear recursion equation $\pi(k+1) = \pi(k)P$ with initial probability distribution $\pi(0) = [1,\:0,\:0,\:0,\:0,\:0]$.
}

{
2) i.i.d. Communication Delays:\\
Although the previous examples in \cite{xiao2000control,zhang2005new} assumed that the communication delays are governed by Markov process, we adopt an i.i.d. jump process to manifestly show that the proposed methods are also applicable to other types of SJLSs.
In case of i.i.d. jump process, the switching probability distribution $\pi(k)$ is stationary, and hence it does not change over time. We assume that the switching probabilities $\pi_{sc}$ and $\pi_{ca}$ are given by
\begin{align*}
\pi_{sc} = [0.7,\: 0.2,\: 0.1],\quad \pi_{ca} = [0.5,\: 0.5]
\end{align*}
where $\pi_{sc}$ and $\pi_{ca}$ stand for the switching probability distribution with respect to sensor-to-controller and controller-to-actuator, respectively. Then, the switching probability $\pi$ for this inverted pendulum system is given by $\pi = \pi_{sc}\otimes \pi_{ca}.$}

Differently from \cite{zhang2005new} where the initial state is deterministically given, we assume that the system contains initial state uncertainties as Gaussian distribution $\mathcal{N}(\mu(0), \Sigma(0))$ with
$\mu(0) = \begin{bmatrix}
0, & 0, & 0.1, & 0
\end{bmatrix}^{\top}$ and $\Sigma(0) = 0.25^2 I_{4\times 4}$, where $I_{4\times 4}$ denotes $4\times 4$ identity matrix.
Moreover, we tested the performance and robustness of this inverted pendulum system with an initial MoG PDF, which is given by
a bimodal Gaussian in the following form
\begin{align*}
\rho(0) = \sum_{j=1}^{2}\alpha_j(0)\mathcal{N}(\mu_j(0), \Sigma_j(0))
\end{align*}
where $\alpha_1(0) = 0.5$ and $\alpha_2(0) = 0.5$. Mean and covariance for each Gaussian component are given by
\begin{align*}
&\mu_1(0) = \begin{bmatrix}
0.5, & 0.25, & -0.12, & 0.05
\end{bmatrix}^{\top},\: \Sigma_1(0) = 0.25^2 I_{4\times 4},\\
&\mu_2(0) = \begin{bmatrix}
-0.4, & 0.35, & 0.07, & -0.1
\end{bmatrix}^{\top},\: \Sigma_2(0) = 0.3^2 I_{4\times 4}.
\end{align*}
These types of multimodal uncertainties are caused by various factors such as sensing under interference\cite{girod2001robust}, distributed sensor networks\cite{langendoen2003distributed}, multitaget tracking problems\cite{liu2007multitarget} and so forth. The bivariate marginal distribution associated with state $x$ and $\theta$ for these Gaussian and MoG PDF are shown in Fig. \ref{fig:4(a)} and Fig. \ref{fig:4(b)}, respectively.

\begin{figure}
\begin{center}
\subfigure[Gaussian marginal distr.]{\includegraphics[scale=0.2]{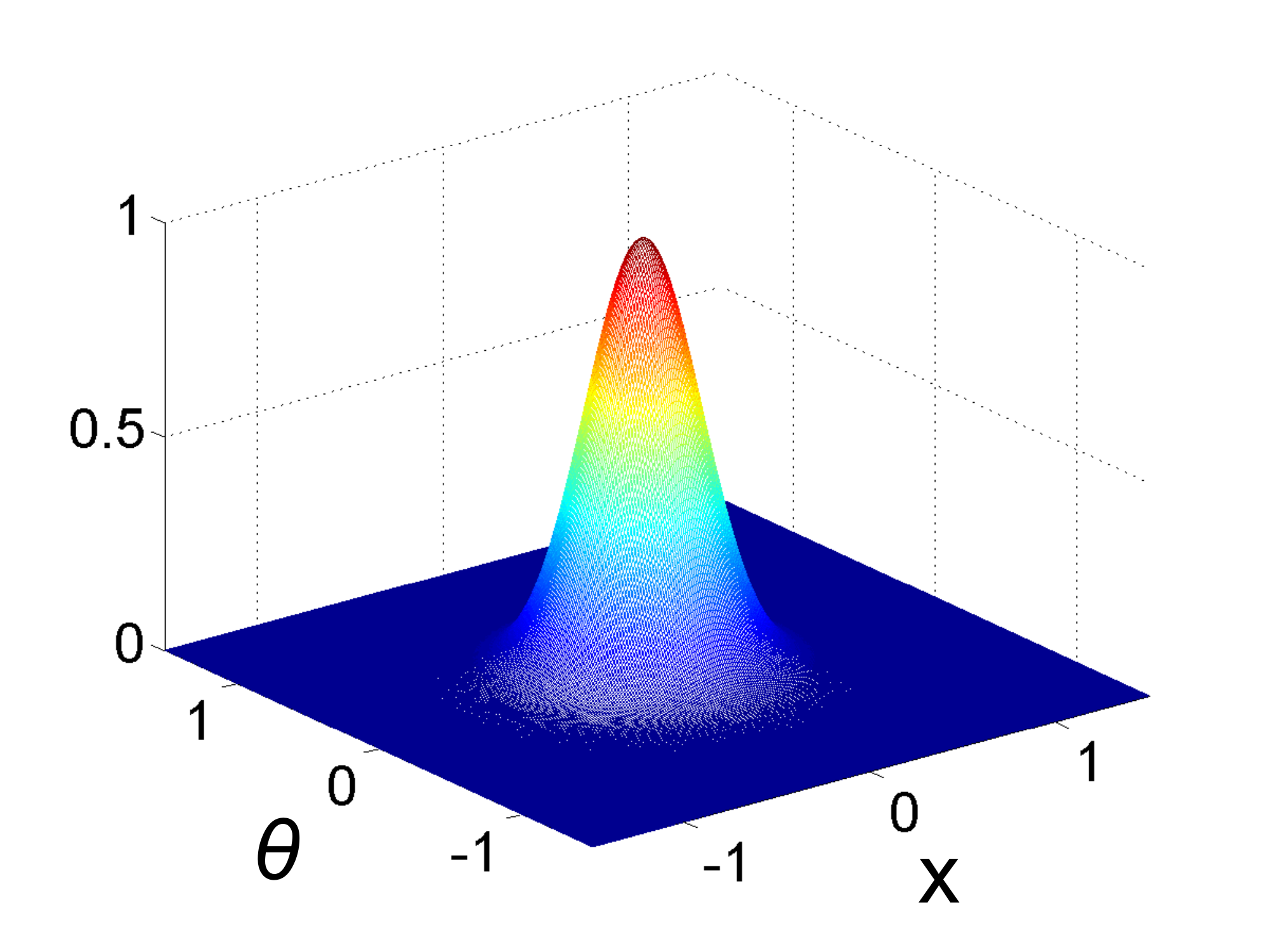}\label{fig:4(a)}}
\subfigure[MoG marginal distr.]{\includegraphics[scale=0.2]{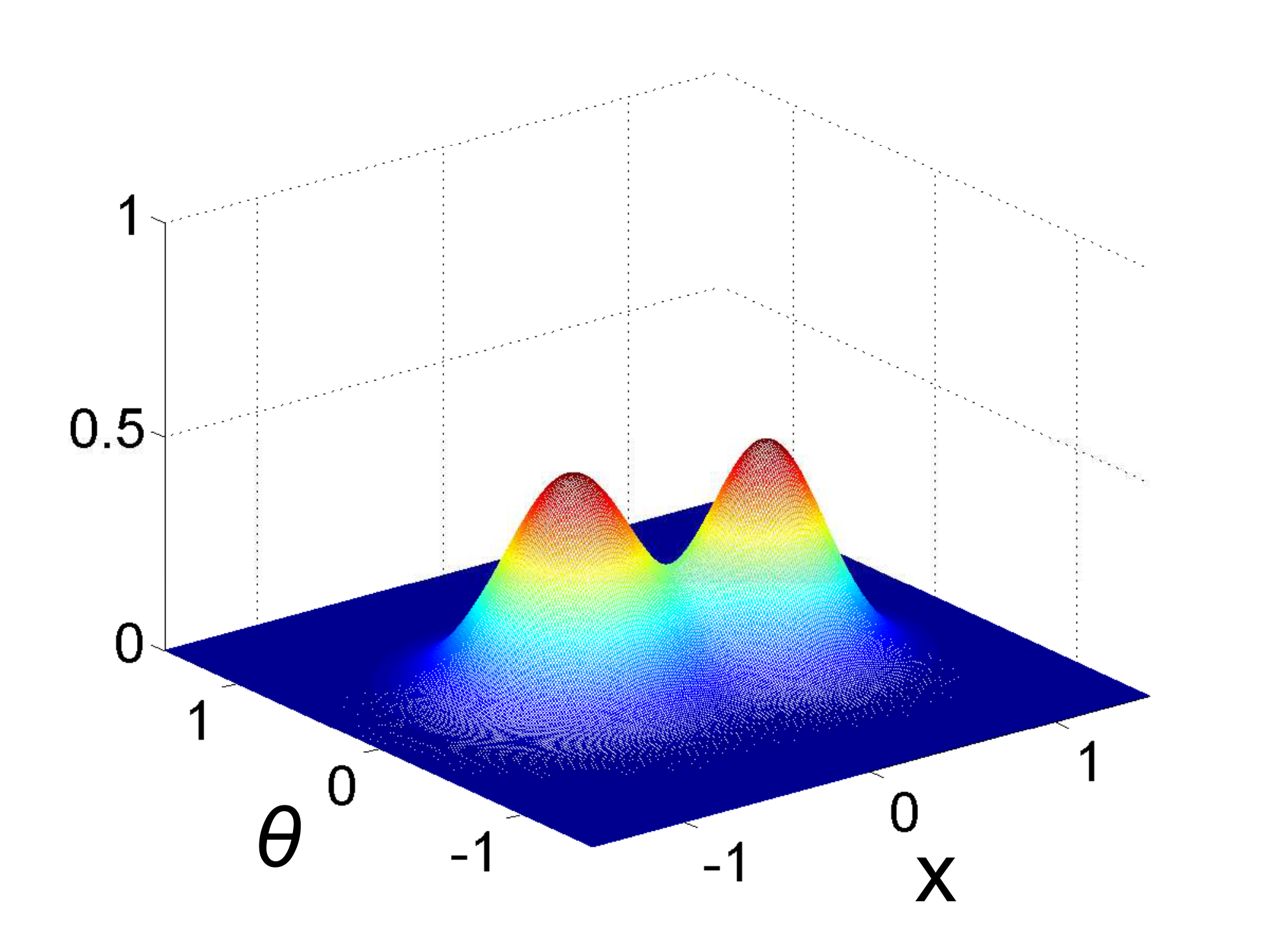}\label{fig:4(b)}}
\subfigure[Wasserstein distance with different stochastic jump processes and initial PDFs; MJLS with Gaussian (blue solid), MJLS with MoG (red dashed), i.i.d. with Gaussian (green triangle), and i.i.d. with MoG (purple cross).]{\includegraphics[scale=0.45]{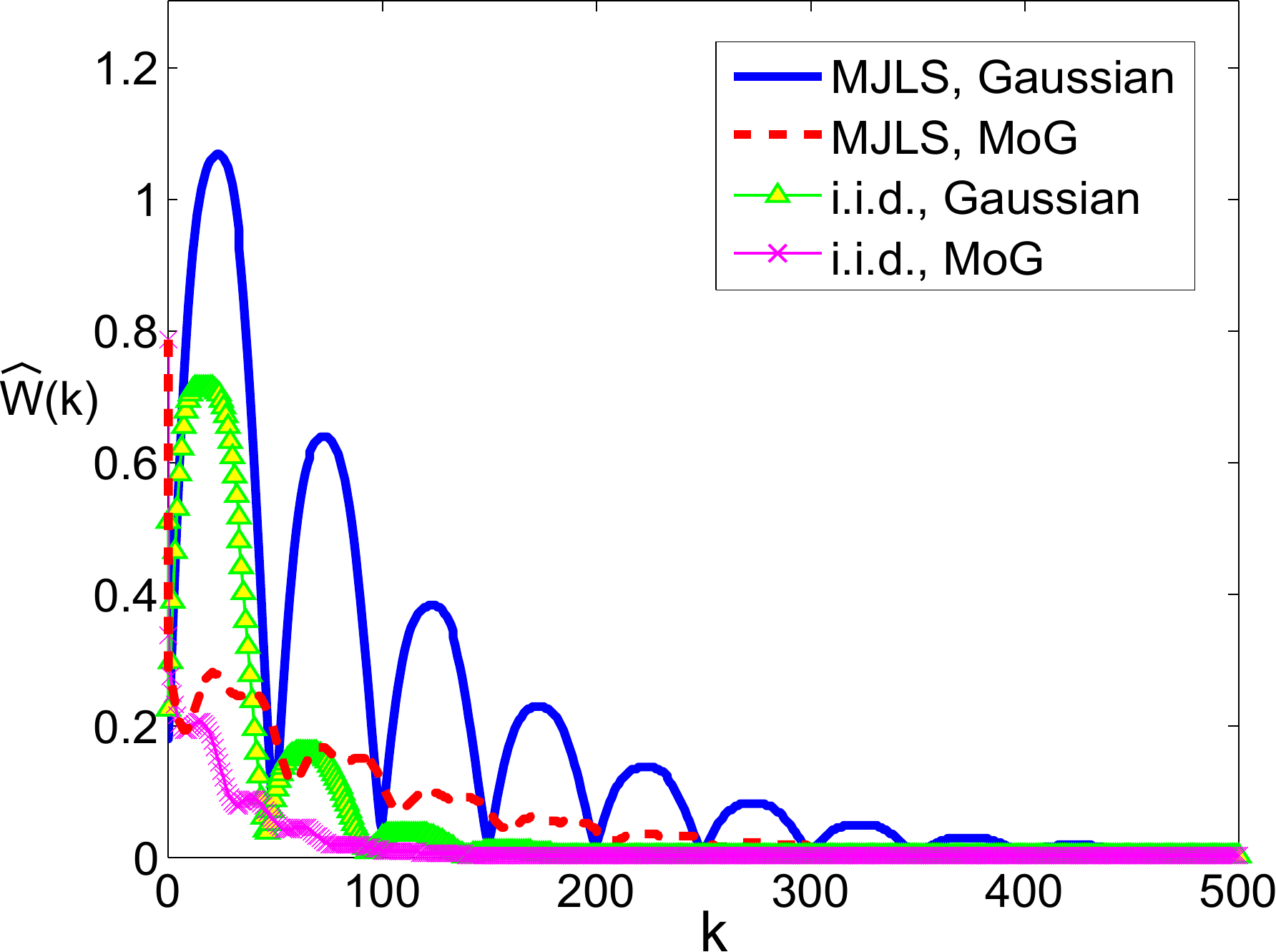}\label{fig:4(c)}}
\caption{Simulation Result for Performance and Robustness Analysis of Inverted Pendulum system with the existence of both random communication delays and initial state uncertainties.}\label{fig_pendulum}
\end{center}
\end{figure}

{
In Fig. \ref{fig:4(c)}, the performance and the robustness of this inverted pendulum system with different stochastic jump processes and initial state uncertainties are depicted via $\widehat{W}$ computation. For all cases, we know that the system is m.s. stable from the convergence of $\widehat{W}$. However, the rate of convergence and the performance show different aspects in the transient time. Among all cases, $\widehat{W}$ for i.i.d. jump process with initial MoG PDF converges fast with small bounce, whereas $\widehat{W}$ for MJLS with initial Gaussian PDF slowly converges with large bounce.
}

At every time step, the ``Split-and-Merge'' algorithm, presented in Section \ref{sec:3.2.2} is used to propagate the state PDFs. Without using these techniques, it is practically impossible to propagate density functions and calculate $W$ {(i.e., the Wasserstein distance between actual state PDF $\rho(x)$ and $\delta(x)$)} even for a finite {switching modes}. The number of Gaussian components that represents the state PDF after $N$ time steps is $6^{N}$, which soon becomes computationally intractable. For an $m$-mode {SJLS}, the growth rate is $m^N$.
{
With the implementation of the proposed ``Split-and-Merge'' algorithm, $\widehat{W}$ that is equidistant with $W$ was computed {effectively and efficiently}. From this example, it is clearly shown that the performance and the robustness for general SJLSs can be measured via $\widehat{W}$ distance which quantifies the uncertainties. }
%The stability of this inverted pendulum system under given initial state uncertainties is also guaranteed by the convergence of $W$.

\section{Conclusion}
{
In this paper, we proposed new tools for the performance and the robustness analysis of stochastic jump linear systems. With given initial state uncertainties, Wasserstein distance that compares shapes of PDFs provides a way to quantify the uncertainties. Since the growth of PDF components in stochastic jumps is exponential in time, we presented a new ``Split-and-Merge'' algorithm for uncertainty propagation that scales linearly with the number of modes in the jump system. This method provides analytical solutions, while avoiding exponential growth of PDF components. The proposed methods are applicable not only to Markovian jumps, which is commonly assumed in the analysis of jump systems, but also to general stochastic jump linear systems. We also proved that mean square stability can be shown with regard to convergence of Wasserstein distance. These results address both transient and steady-state behavior of stochastic jump linear systems. The practical usefulness and efficiency of the proposed method are verified by examples.}
\bibliographystyle{plain}
\bibliography{Automatica}

\end{document}